\documentclass{sig-alternate-per}

\usepackage{tikz}

\def \bi {\begin{itemize}\item}
\def \ei {\end{itemize}}
\def \be {\begin{equation}}
\def \ee {\end{equation}}
\def \ba {\begin{aligned}}
\def \ea {\end{aligned}}
\def \non {\nonumber}

\begin{document}
\newtheorem{theorem}{Theorem}
\newtheorem{lemma}{Lemma}
%

\title{Network Non-Neutrality on the Internet: Content Provision Under a Subscription Revenue Model}

%

\numberofauthors{3}

\author{
\alignauthor
Mohammad Hassan Lotfi\\
\affaddr{ESE Department}      \\ 
       \affaddr{University of Pennsylvania}\\
       \affaddr{Philadelphia, PA, 19104}\\
       \email{lotfm@seas.upenn.edu}
\alignauthor
George Kesidis\\
       \affaddr{EE and CS\& E Department}\\
       \affaddr{Pennsylvania State University}\\
       \affaddr{University Park, PA, 16803 }\\
       \email{kesidis@engr.psu.edu}
\alignauthor 
Saswati Sarkar\\
      \affaddr{ESE Department}      \\ 
       \affaddr{University of Pennsylvania}\\
       \affaddr{Philadelphia, PA, 19104}\\
       \email{swati@seas.upenn.edu}
}

\maketitle
\begin{abstract}
The goal of this paper is to provide an insight into the equilibrium of the Internet market, when the current balance of the market is disrupted, and one of the ISPs switches to a non-neutral regime. We consider a content provider with a subscription revenue model and a continuum of end-users. The CP is also non-neutral, in the sense that she can charge users of different ISPs different subscription fees, and use this ``leverage" to control the equilibrium outcome. Results reveal that the CP is able to control the non-neutral ISP to some extend. However, switching to a non-neutral regime by an ISP tips the balance of the market in favor of this ISP. 
\end{abstract}




\section{Introduction}
Net neutrality on the Internet is perceived as the policy that mandates Internet Service Providers (ISPs) to treat all data equally, regardless of the source, destination, and type of the data \cite{progressive}. Recently, the net-neutrality debate has received more attention, since in January 2014, a federal appeals court struck down parts of the Federal Communication Commission's rules for Net-Neutrality \cite{NYtimes1}. Subsequently, Comcast and Netflix signed an agreement in February 2014 in which Netflix should pay Comcast for a faster access to Comcast's subscribers \cite{NYtimes2}. 

However, changing the balance of the Internet market by ISPs and switching to a non-neutral regime may eventually lead to a battle between ISPs and Content Providers (CPs). Depending on their power over the market, the winner can be any of the two. For instance, one can think about scenarios that ISPs end up paying a popular CP to have the privilege of carrying its data. The goal of this paper is to model the interaction between ISPs and CPs in a non-neutral regime. 

Most of the literature about the Net Neutrality debate fall into the realm of law and policy making. A survey of the existing scarce literature on economics analysis of the net-neutrality debate is presented in \cite{survey}. In this genre of work, the social welfare analysis of the neutral and non-neutral regime has been taken into account largely, while the results and conclusions vary to a great extend depending on the model used. In some cases, non-neutrality increases the social welfare because of the increase in the investments by ISPs. On the other hand, some results imply a lower social welfare in a non-neutral regime than a neutral one. 

In this paper, we try to get an insight into the equilibrium of the Internet market, when one of the ISPs switches to non-neutrality. Elements of the market are two ISPs, a CP with a subscription revenue model, and a continuum of end-users. In our model, the CP should pay a per-subscriber fee to the non-neutral ISP in order to access the ISP's subscribers. In addition, the CP can potentially charge users of different ISPs with different subscription fees. This is what we call the \emph{``leverage"} of the CP over end-users and subsequently ISPs, by which the CP can control the equilibrium of the market.  

We use the standard two-sided market framework introduced in \cite{armstrong}, and used in the context of the Internet market in \cite{economides} (Section~\ref{section:model}). We seek the sub-game perfect equilibrium of the sequential game using backward induction (Section~\ref{section:SPE}).  The results imply that in spite of the leverage of the CP  over ISPs and end-users, switching to non-neutrality will tip the balance of the network in favor of the non-neutral ISP. More discussions on the equilibrium of the market and its implications are presented in Section~\ref{section:discussion}.  

\section{Market Model}\label{section:model}

We cast the problem with a two-sided market framework which captures the Internet market appropriately. In addition, the problem is modelled as a sequential game. Two Internet Service Providers (ISPs) are considered that act as platforms that connect the two sides of the Internet market: a Content Provider (CP) that charges end-users with a subscription fee and a continuum of end-users.

\subsubsection*{The ISPs}

ISPs provide connection between CPs and end-users. We assume that one of the ISPs is neutral (ISP N) and the other is non-neutral (ISP NN). ISP N and NN provide Internet connection with the same quality to end-users in exchange of a subscription fee, $p_N, p_{NN}\in \mathbb{R}$, respectively.  Once the subscription fee is paid by an end-user to the neutral ISP, ISP N provides the connection between the end-user and the CP. However, ISP NN offers the CP a take-it-or-leave-it per-subscriber fee, $\tilde{p}$. This means that if the CP takes the offer, she should pay a fixed fee $\tilde{p}\in \mathbb{R}$ per each subscriber she gets from those end-users that are connected to ISP NN. If not, ISP NN will block the access of the CP to her end-users.


\begin{table*}[t]
\small
\centering
    \begin{tabular}{ | p{5.5cm} || p{10cm} |}
    \hline
    Symbol & Description \\ \hline \hline
    $p_N$  & Internet connection fee charged by ISP N \\ \hline 
    $\tilde{p}$ & per-subscriber fee that G pays to ISP NN\\ \hline
    $q_N$ & subscription fee that G charges ISP N's users \\ \hline
	$\pi_N$ & payoff of ISP N\\ \hline
	$n_{C_{N}}$ & fraction of users that buy Internet from ISP N\\ \hline
	$n'_{C_N}$ & fraction of users that buy Internet from ISP N and pay for the content of G\\ \hline
	$u_{j,I}(x)$ & payoff from accessing to Internet for a user located at $x$ connecting to ISP j \\ \hline
	$u_{j,G}(x)$ & payoff from content of G for a user located at $x$ connecting to ISP j\\ \hline
	$u_{j}(x)$ & total payoff of a user located at $x$ connecting to ISP j\\ \hline
	$\pi_G$ & payoff of the CP G\\ \hline  
  $p_{NN}$, $q_{NN}$,  $\pi_{NN}$, $n_{C_{NN}}$, $n'_{C_{NN}}$ & corresponding parameters for ISP NN\\ \hline
	$\Delta q$ & $q_{NN}-q_N$\\ \hline
	$p^e_N$, $p^e_{NN}$, $\tilde{p}^e$, $q^e_N$, $q^e_{NN}$, $n^e_{C_N}$, $n^e_{C_{NN}}$, $\Delta q^e$ & corresponding parameters at the equilibrium\\ \hline
    \end{tabular}
    \caption{Important Symbols }\label{table:symbols}
\end{table*}

\normalsize

\subsubsection*{The CP}

We assume that there exists only one CP, and denote this CP by G.  However, the results of this paper can be generalized to a continuum of independent monopolist CPs that do not compete with each other. 

In our model, CP G makes profit by a subscription-based model and can multi-home, i.e. can offer her content via both ISPs . In addition, G is considered to be non-neutral, in the sense that she can charge users of different ISPs differently. This strategic behavior of the CP can be interpreted as a leverage by which CP G can potentially decrease the incentive of an ISP for switching to a non-neutral regime. We denote the subscription fees that G charges users of non-neutral and neutral ISP by $q_N, q_{NN}\in \mathbb{R}$, respectively.


\subsubsection*{The End-Users}

We assume a continuum of single-homing customers, i.e. they can connect to the Internet via only one ISP. We model the valuation of customers for ISPs and CP G using a hotelling model. In our hotelling model, we assume that the neutral ISP is located at 0, the non-neutral one is located at 1, and end-users are distributed uniformly along the unit interval, $[0,1]$. 

Let $v$ denote the common valuation of end-users for the Internet connection. The overall valuation of an end-user located at $x\in[0,1]$ for connecting to the Internet via the neutral and non-neutral ISP is $v-tx$  and $v-t(1-x)$, respectively, where  $tx$ and $t(1-x)$ are the transport costs and $t$ is the marginal transport cost. In words, the closer an end-user to an ISP, the more the end-user prefers this ISP to the other\footnote{Note that the distance between an ISP and an end-user defined in this model is different from the physical distance.}.

One possible interpretation for the usage of the transport cost is to capture the inertia of end-users to change their current ISP due to  their high budget and their prior experience with the ISP. In other words, wealthier users with high budget are less sensitive to the price they pay, more reluctant to change their ISP, and subsequently  located near one of the ISPs depending on their prior experience with ISPs. Thus these users have higher valuation for connecting to the Internet via the closest ISP.  

In addition, we assume that the valuation of a user for content of G increases by her budget. In other words, the less sensitive a user to the price, the higher the chance of buying the content from G. Since we assumed that users with higher budgets are located near one of the ISPs, the closer an end-user to one of the ISPs, i.e. points 0 or 1 on the unit interval, the higher the valuation of the user for CP G. Let $v^*$ denote the common valuation of end-users for CP G. A possible expression for the overall valuation of an end-user located at $x\in[0,1]$ for CP G is $v^*-t\min\{x,1-x\}$.


Figure~\ref{figure:two sided} illustrates a schematic view of the two sided market we consider. A list of more important symbols, used throughout the paper, is presented in Table~\ref{table:symbols}.

\begin{figure}[t]
\centering
\includegraphics[width=0.45\textwidth]{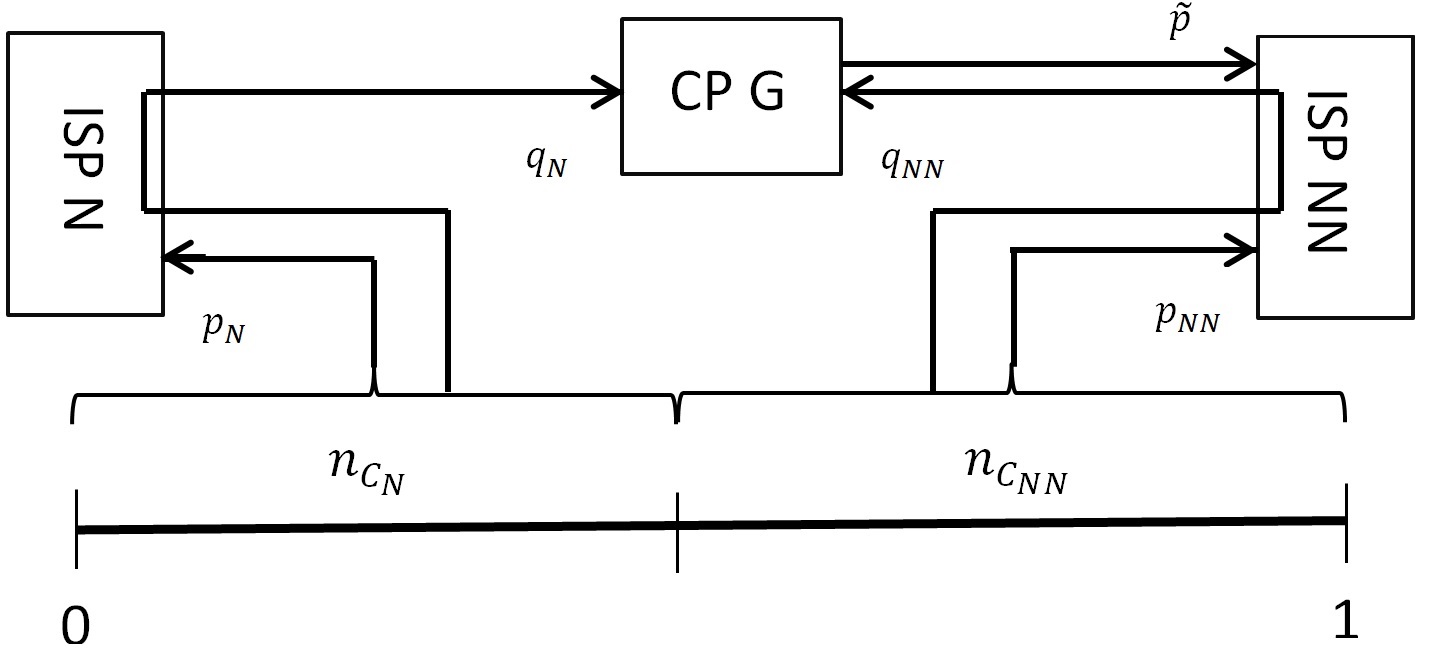}
\caption{Schematic view  of our two sided market }
\label{figure:two sided}
\end{figure}

\subsubsection*{Payoffs}

The payoffs obtained by ISP NN and N are:

\begin{eqnarray}\label{equation:utilityISP}
\begin{aligned}
\pi_{NN}&=(p_{NN}-c)n_{C_{NN}}+\tilde{p}n'_{C_{NN}}\\
\pi_N&=(p_N-c)n_{C_N}
\end{aligned}
\end{eqnarray}
where $c$ is the per connection cost; $n_{C_N}$, $n_{C_{NN}}$, and  $n'_{C_{NN}}$ are defined in Table~\ref{table:symbols}.

At the end-user side, the payoff of a user that is located at distance $y$ of ISP $j\in\{N,NN\}$ consists of two parts: the payoff from the Internet connection, and the payoff from accessing to the content of G:
\be \label{equ:UtilityFromI}
u_{j,I}(y)=v-ty-p_j
\ee
\be\label{equ:UtilityFromG}
u_{j,G}(y)=v^*-t\min\{y,1-y\}-q_j 
\ee
 Thus, the payoff of the end-user is,
\be\label{equation:CP_2}
\ba
u_j(y) =
\left\{
	\begin{array}{ll}
		u_{j,I}(y)+u_{j,G}(y)  &\mbox{Pays for Internet and G}\\
	u_{j,I}(y) &\mbox{Pays for Internet and not G}\\
0 &\mbox{Does not pay for Internet}
	\end{array}
\right.
\ea 
\ee





At the CP side, the payoff of G is,
\begin{equation}\label{equation:utilityCP}
\pi_{G}=(q_{NN}-\tilde{p})n'_{C_{NN}}+q_Nn'_{C_N}
\end{equation}

\subsubsection*{Assumptions about the Market}

We assume the full market coverage by ISPs. In other words, $p_N$ and  $p_{NN}$ are chosen by ISPs to ensure that all end-users are willing to pay for the Internet connection to at least one ISP. More formally, at the equilibrium, for $x\in[0,1]$,
\be \label{equ:fullcoverage}
u_j(x)\geq 0 \quad \text{ for at least one } j\in\{N,NN\}
\ee

We argue in Section~\ref{section:discussion} that $v\geq 2t+c$ is a sufficient condition for full market coverage at the equilibrium. Note that the full coverage by ISPs is a natural and common assumption. This assumption bundled with hotelling model is necessary to create an element of competition between ISPs. 

In addition, when solving the problem of CP in Section~\ref{section:SPE}, we assume that CP G chooses $q_N$ and $q_{NN}$ such that all end-users are willing to pay the subscription fee after connecting to the Internet. Thus,  $n'_{C_{NN}}=n_{C_{NN}}$ and $n'_{C_{N}}=n_{C_{N}}$. More formal condition for the market coverage by $G$ is that for $x\in[0,1]$,
$$
u_{j,G}(x)\geq 0 \quad \text{for ISP j that user selects}
$$
The full market coverage by G arises because of the fact that for CP G, not only the subscription fee, but also the number of subscribers is important. This may happen because of connectivity considerations (e.g. social networks with subscription revenue model that obviously care about maximum connectivity), or financial considerations (aiming to fetch some advertisement-based revenue). More discussion about the implications of this assumption is deferred to Section~\ref{section:discussion}.

  
\subsubsection*{Timing}

The timing of the sequential game is modeled as follows:
\begin{enumerate}
\item ISPs announce the per-subscriber fee for the CP G ($\tilde{p}$).
\item G sets the subscription fees ($q_N$ and $q_{NN}$).
\item ISPs set the connection fees ($p_N$ and $p_{NN}$).
\item End-users decide whether to participate and, if so, which ISP to join.
\end{enumerate}
We assume that the per-subscriber fee is set before the connection fees by ISP NN in order to capture the longer time horizon of contracts between ISPs and CPs. In addition, to make the game between ISPs and the CP more balanced, it is assumed that ISPs set the connection fees after G sets the subscription fees. 

\section{The Sub-Game Perfect Equilibrium}\label{section:SPE}
In this section, we seek the sub-game perfect equilibrium of the sequential game using backward induction.  Thus, we start with the last stage:

\subsubsection*{Stage 4: End-users decide:}

First note that the full market coverage by CP G yields that $p_N$, $p_{NN}$, $q_N$, and $q_{NN}$ are determined such that end-users always pay for the connection and the content of G. Thus in \eqref{equation:CP_2}, $u_j(y)=u_{j,I}(y)+u_{j,G(y)}$. The customer located at $x_n\in[0,1]$ is indifferent between connecting to ISP 1 and 2:

$$
\begin{aligned}
v-&t(1-x_n)-p_{NN}+v^*-t\min\{1-x_n,x_n\}-q_{NN}\\
&=v-tx_n-p_N+v^*-t\min\{x_n,1-x_n\}-q_N\\
\end{aligned}
$$ 
\begin{equation}\label{equation:equ4}
\Rightarrow x_n=n^e_{C_N}=\frac{1}{2}+\frac{p_{NN}-p_{N}+q_{NN}-q_{N}}{2t}
\end{equation}
and the full market coverage implies that  $n^e_{C_{NN}}=1-n^e_{C_N}$.

\subsubsection*{Stage 3: ISPs set $p_N$ and $p_{NN}$:}

 Utilizing the equilibrium outcome of the fourth stage \eqref{equation:equ4}, payoffs obtained by the neutral and non-neutral ISPs are:

\begin{equation}
\begin{aligned}
\pi_{NN}&=(p_{NN}+\tilde{p}-c)(\frac{1}{2}-\frac{p_{NN}-p_{N}+q_{NN}-q_{N}}{2t})\\
\pi_N&=(p_N-c)(\frac{1}{2}+\frac{p_{NN}-p_{N}+q_{NN}-q_{N}}{2t})
\end{aligned}
\end{equation}

Note that $\pi_{NN}$ and $\pi_{N}$ are concave functions of $p_{NN}$ and $p_N$, respectively.  Applying the first order condition to payoffs ($\frac{d\pi_N}{dp_{NN}}=0$ and $\frac{d\pi_{N}}{dp_N}=0$) yields the equilibrium strategies, $p^e_{NN}$ and $p^e_{N}$: 

\begin{equation}\label{equation:price_customers}
\begin{aligned}
p^e_{NN}&=t+c-\frac{q_{NN}-q_N}{3}-\frac{2\tilde{p}}{3}\\
p^e_{N}&=t+c+\frac{q_{NN}-q_N}{3}-\frac{\tilde{p}}{3}
\end{aligned}
\end{equation}

Note that $\tilde{p}$ is known from stage $1$ (to be explained later). Thus, it is not explicitly dependent on $p_N$ and $p_{NN}$. By substitution in \eqref{equation:equ4}, the equilibrium outcome of the fourth stage can be re-written as,
\begin{equation}\label{equation:customers}
\begin{aligned}
n^e_{C_N}&=\frac{1}{2}+\frac{q_{NN}-q_N-\tilde{p}}{6t}\\
n^e_{C_{NN}}&=\frac{1}{2}-\frac{q_{NN}-q_N-\tilde{p}}{6t}
\end{aligned}
\end{equation}

\subsubsection*{Stage 2: G sets $q_N$ and $q_{NN}$:}

We first state the conditions for the full coverage of the end-user market by CP G.  Let $\hat{x}_N$ (respectively, $\hat{x}_{NN}$) denote the location of the customer with the largest distance from N (resp., NN), which opts to buy the content from G if she joins ISP N (resp., NN). Using the valuation of users for G, for the full market coverage, $q_N$ and $q_{NN}$ should be chosen such that:

\begin{figure}[t]
\begin{center}
\resizebox{8.5cm}{1.3cm}{
\begin{tikzpicture}
\draw [ultra thick] (0,2) -- (8,2);
\draw [thick] (0,1.9) -- (0,2.1)  node[above]{ISP N};
\draw [thick] (8,1.9) -- (8,2.1) node[align=left,above]{ISP NN} ;
\draw [thick] (3,1.9)  -- (3,2.1) node[above]{$\hat{x}_{NN}$};
\draw [thick] (5,1.9) -- (5,2.1)  node[above]{$\hat{x}_N$};
\draw [thick] (4.2,1.9) -- (4.2,2.1)  node[above]{$n^e_{C_N}$};
\draw [thick] [<->] (4.2,1.8) --node[below]{\scriptsize{$n^e_{C_{NN}}$}}(8,1.8) ;
\draw [thick] [<->] (0,1.8) --node[below]{\scriptsize{$n^e_{C_{N}}$}}(4.2,1.8) ;
\end{tikzpicture}
}
\end{center}
\caption{Conditions for full Market Coverage}\label{figure:structure}
\end{figure}
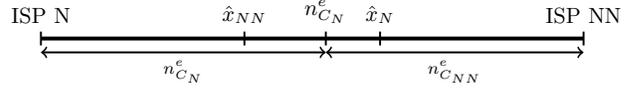

\begin{itemize}
\item For the neutral ISP:
\ei

- If $n^e_{C_N}\geq\frac{1}{2}$: From~\eqref{equ:UtilityFromG}, $v^*-\frac{t}{2}-q_N\geq 0 \Rightarrow q_N\leq v^*-\frac{t}{2}$ . In this case, all the customers in the unit interval are willing to buy the subscription plan. 

- If $n^e_{C_N}<\frac{1}{2}$:  $\hat{x}_N\geq n^e_{C_N}$ (Figure~\ref{figure:structure}), where $\hat{x}_N$ is such that $v^*-t\hat{x}_N-q_N = 0$. Thus, $\hat{x}_N=\frac{v^*-q_N}{t}\geq n_{C_N}$.

\bi  For the non-neutral ISP:
\ei 

-If $n^e_{C_N}<\frac{1}{2}$:  $v^*-\frac{t}{2}-q_{NN}\geq 0 \Rightarrow q_{NN}\leq v^*-\frac{t}{2}$ .

- If  $n^e_{C_N}\geq \frac{1}{2}$:  $\hat{x}_{NN}\leq n^e_{C_N}$ (Figure~\ref{figure:structure}), where $\hat{x}_{NN}$ is such that $v^*-t(1-\hat{x}_{NN})-q_{NN} = 0$. Thus, $\hat{x}_{NN}=\frac{t+q_{NN}-v^*}{t}\leq n_{C_N}$.

Taking $\Delta q=q_{NN}-q_N$ yields that $n_{C_N}\geq\frac{1}{2}$ (resp., $n_{C_N}<\frac{1}{2}$) IFF $\Delta q \geq \tilde{p}$ (resp., $\Delta q < \tilde{p}$). Thus, G  determines $q_N$ and $q_{NN}$ by solving the following optimization problem:

\begin{equation}\label{equ:maximization}
\max \pi_{G}=(q_{NN}-\tilde{p})(\frac{1}{2}-\frac{\Delta q-\tilde{p}}{6t})+q_N(\frac{1}{2}+\frac{\Delta q-\tilde{p}}{6t})
\end{equation}

s.t:

\begin{equation}\label{equ:const1}
\left\{
	\begin{array}{ll}
		 q_N\leq v^*-\frac{t}{2}=u_{q_N} \\
		q_{NN}\leq v^*-\frac{t}{2}+\frac{\Delta q-\tilde{p}}{6}=u_{q_{NN}}
	\end{array}
\right.    \text{if} \quad  \Delta q \geq \tilde{p}
\end{equation}

\begin{equation}\label{equ:cons2}
\left\{
	\begin{array}{ll}
		 q_N\leq v^*-\frac{t}{2}-\frac{\Delta q-\tilde{p}}{6}=u'_{q_N} \\
		q_{NN}\leq v^*-\frac{t}{2}=u'_{q_{NN}}
	\end{array}
\right.    \text{if} \quad \Delta q \leq  \tilde{p}
\end{equation}

The solution is presented in the following theorem. Theorem~\ref{theorem:eq2} is proved in the Appendix.

\begin{theorem}\label{theorem:eq2}
The equilibrium strategy of the CP G depends on $\tilde{p}$ and is as follows:

\begin{itemize}
\item If $\tilde{p}\leq -\frac{5t}{4}$, then $\Delta q^e=\tilde{p}+\frac{3t}{2}$, $q^e_{N}=u_{q_N}=v^*-\frac{t}{2}$, and $q^e_{NN}=q^e_N+\Delta q^e$.
\item If $ -\frac{5t}{4} \leq \tilde{p} \leq 0$, then $\Delta q^e=-\frac{\tilde{p}}{5}$, $q^e_{N}=u_{q_N}=v^*-\frac{t}{2}$, and $q^e_{NN}=q^e_N+\Delta q^e$.
\item If  $ 0 \leq \tilde{p} \leq \frac{5t}{4}$, then $\Delta q^e=-\frac{\tilde{p}}{5}$, $q^e_{NN}=u'_{q_{NN}}=v^*-\frac{t}{2}$, and $q^e_{N}=q^e_{NN}-\Delta q^e$.
\item If $\tilde{p}\geq \frac{5t}{4}$, then $\Delta q^e=\tilde{p}-\frac{3t}{2}$, $q^e_{NN}=u'_{q_{NN}}=v^*-\frac{t}{2}$, and $q^e_{N}=q^e_{NN}-\Delta q^e$.
\end{itemize} 
\end{theorem} 

\subsubsection*{Stage 1: The Non-Neutral ISP sets $\tilde{p}$:}

Using the equilibrium strategies obtained in stages 2 to 4, the payoff of the non-neutral ISP depending on the range of $\tilde{p}$ she chooses is:

\be
\ba
\pi^e_{NN} =
\left\{
	\begin{array}{ll}
		\frac{t}{8}  &\mbox{if } \tilde{p}\leq -\frac{5t}{4}\\
	(t+\frac{2}{5}\tilde{p})(\frac{1}{2}+\frac{\tilde{p}}{5}) &\mbox{If } -\frac{5t}{4} \leq \tilde{p} \leq \frac{5t}{4}\\
\frac{9t}{8} &\mbox{If } \tilde{p}\geq \frac{5t}{4}
	\end{array}
\right.
\ea 
\non
\ee
Therefore at the equilibrium, $\tilde{p}^e \geq \frac{5t}{4}$.

\section{Discussion}\label{section:discussion}

We first find a sufficient condition for the full market coverage by ISPs. At the equilibrium, $\tilde{p}^e\geq \frac{5t}{4}$, $p^e_{NN}=\frac{3t}{2}+c-\tilde{p}^e$, and $p^e_{N}=c+\frac{t}{2}$. Thus, the maximum price that ISPs can charge users occurs when $\tilde{p}^e=\frac{5t}{4}$, and is less than $t+c$. Thus, from \eqref{equ:UtilityFromI} and \eqref{equ:fullcoverage}, the assumption $v>2t+c$ is a sufficient condition for the full market coverage. 

Note that $q^e_{N}$ and $q^e_{NN}$ are adjusted by G such that the payoff of ISP NN for $\tilde{p} \geq \frac{5t}{4}$ is independent of $\tilde{p}$.  Thus, the non-neutral ISP cannot strictly increase her payoff by charging CP G a very high price. This is the result of the leverage of the CP G over ISPs and end-users. 

In addition, note that at the equilibrium, $n^e_{C_{NN}}>n^e_{C_{N}}$, i.e. the non-neutral ISP attracts more customers, and $p^e_{NN}<p^e_{N}$ which implies that the non-neutral ISP subsidizes end-users with the money that she collects from $G$. This increases the differentiation between platforms, and makes the non-neutral ISP more favorable for end-users. 

Another interesting result is that  $\Delta q^e=\tilde{p}^e-\frac{3t}{2}$ can be positive or negative at the equilibrium, depending on $\tilde{p}^e$. If $\tilde{p}^e=\frac{5t}{4}$, then $\Delta q^e<0$. This means that CP G charges the users of the non-neutral ISP with a price lower than the users of the neutral one, which is counter-intuitive. To see what derives  this counter-intuitive result, first note that from \eqref{equation:customers}, $n_{C_{NN}}=\frac{1}{2}-\frac{\Delta q-\tilde{p}}{6t}$. Thus, if ISP NN be able to decrease $\Delta q -\tilde{p}$ by choosing $\tilde{p}$ high enough, she can obtain more per-subscriber revenue from G and also increase the number of her end-users. The later is because of the fact that ISP NN subsidizes end-users from the money she gets from G  ($p^e_{NN}=\frac{3t}{2}+c-\tilde{p}$). Thus, the higher $\tilde{p}$, the higher the subsidies, and so the number of end-users that choose  ISP NN for Internet connection. This leads to the desire of ISP NN to choose $\tilde{p}$ in the equilibrium such that $n^e_{C_{NN}}>\frac{1}{2}$, and force the equilibrium solution of CP G to satisfy constraints \eqref{equ:cons2} rather than \eqref{equ:const1}. In this case, G maximizes her payoff by increasing the upper bound on $q_N$\footnote{Note that the upper bound for $q_{NN}$ is fixed in these constrains.}, which she do so by choosing $\Delta q-\tilde{p}$ to be negative. This leads to a higher upper bound for $q_N$ in comparison to $q_{NN}$, which leads to our counter-intuitive, yet important, result.  

An alternative interpretation for this result is that given that ISP NN attracts more than half of the end-users by appropriately choosing $\tilde{p}$ and subsidizing the end-users thereafter, G knows that those who stay with the neutral ISP are those users who are less sensitive to price and have a higher budget. Therefore G compensates what she pays to ISP NN by charging these users a higher price. 

The results imply that in spite of the leverage of CP G in determining her subscription fees, switching to non-neutrality by an ISP will tip  the balance of the Internet market in favor of this ISP. In fact, the non-neutral ISP can extract some profits from CP G as well as the neutral ISP. The results of this model are dependent on the assumption of the full market coverage by G, since this assumption may reduce the strength of the leverage of G.

A direction for future works is to investigate the market equilibrium when G is not restricted to cover the end-user market. Another possible  direction, is to consider investment decisions by ISPs.

\appendix
\section{Proof of Theorem 1 }

In order to prove the Theorem, we first narrow down the candidate answers in Lemma \ref{lemma:upper} and \ref{lemma:sets}. 

\begin{lemma}\label{lemma:upper}
In an optimum solution of \eqref{equ:maximization}-\eqref{equ:cons2}, at least one of $q_N$ or $q_{NN}$ should be equal to its upper bound.
\end{lemma}


\begin{proof}
Suppose not. Thus, the optimum $q_N$ and $q_{NN}$ are both strictly smaller than their upper bounds. Consider a sufficiently small $\epsilon>0$. Note that $q_N+\epsilon$ and $q_{NN}+\epsilon$ are feasible solutions. Using (\ref{equ:maximization}), it can be concluded that $\pi_G(q_N+\epsilon,q_{NN}+\epsilon)>\pi_G(q_N,q_{NN})$.   This is a contradiction with $q_N$ and $q_{NN}$ being optimum solutions.  
\end{proof}

\begin{lemma}\label{lemma:sets}
There are four sets of candidate answers to \eqref{equ:maximization}-\eqref{equ:cons2}, depending on $\tilde{p}$: 
\begin{enumerate}
\item $  -5 \Delta q  \leq  \tilde{p} \leq \Delta q$ and $q_{NN}=u_{q_{NN}}$.
\item  $ \tilde{p} \leq \min\{\Delta q, -5 \Delta q\}$ and $q_{N}=u_{q_{N}}$.
\item $\tilde{p}\geq \max\{\Delta q,-5 \Delta q \}$ and  $q_{NN}=u'_{q_{NN}}$. 
\item  $ \Delta q \leq \tilde{p} \leq -5 \Delta q $ and $q_N=u'_{q_N}$.
\end{enumerate}
\end{lemma}
\begin{proof}
First note that if $\Delta q\geq \tilde{p}$, then  $u_{q_{NN}}\geq u_{q_N}$. If $q_{NN}=u_{q_{NN}}$, then the constraint for $q_{N}$ should be satisfied:

$$
\ba 
q_N=u_{q_{NN}}-\Delta q &\leq u_{q_N}\Rightarrow \Delta q\geq u_{q_{NN}}-u_{q_N}\\
&\Rightarrow  \Delta q \geq \frac{\Delta q-\tilde{p}}{6} \Rightarrow \tilde{p} \geq -5 \Delta q 
\ea
$$
Thus, $ \Delta q \geq \tilde{p} \geq -5 \Delta q $, which is the first candidate set. 

On the other hand, if $q_N=u_{q_N}$, then the constraint for $q_{NN}$ should be satisfied:
$$
\ba 
q_{NN}=u_{q_N}+\Delta q \leq u_{q_{NN}}&\Rightarrow u_{q_{NN}}-u_{q_N}\geq \Delta q\\
& \Rightarrow \tilde{p} \leq -5 \Delta q 
\ea
$$
Thus, $ \tilde{p} \leq \min\{\Delta q, -5 \Delta q \}$.

In another case, if  $\Delta q\leq \tilde{p}$, then  $u'_{q_{N}}\geq u'_{q_{NN}}$. If $q_{NN}=u'_{q_{NN}}$ then,
$$
\ba 
q_N=u'_{q_{NN}}-\Delta q \leq u'_{q_N} &\Rightarrow \Delta q\geq u'_{q_{NN}}-u'_{q_N}\\
&\Rightarrow \tilde{p} \geq -5\Delta q 
\ea 
$$
Thus, $\tilde{p}\geq \max\{\Delta q,-5\Delta q \}$.

On the other hand, if $q_N=u'_{q_N}$ then,

$$
\ba 
q_{NN}=u'_{q_N}+\Delta q \leq u'_{q_{NN}}&\Rightarrow u'_{q_{NN}}-u'_{q_N}\geq \Delta q\\
&\Rightarrow \tilde{p} \leq -5 \Delta q 
\ea
$$
Thus, using $\Delta q\leq \tilde{p}$, $ \Delta q \leq \tilde{p} \leq -5\Delta q $. The proof is complete.

\end{proof}

In order to find the best response, first consider the first set of candidate answers. Thus, $q_{NN}=u_{q_{NN}}=v^*-\frac{t}{2}+\frac{\Delta q-\tilde{p}}{6}$, $q_N=q_{NN}-\Delta q$, and subsequently equation (\ref{equ:maximization}) is only a function of  $\Delta q$.  Note that (\ref{equ:maximization}) is concave with respect to $\Delta q$, and, 
$$
\begin{aligned}
\frac{d q_{NN}}{d \Delta q}&=\frac{1}{6}\\
\frac{d q_N}{d \Delta q}&=-\frac{5}{6}
\end{aligned}
$$

Solving $\frac{d\pi_G}{d\Delta q}=0$ yields that $\Delta q^*=\tilde{p}-t$. Therefore, $\Delta q^*<\tilde{p}$, which does not satisfy the condition for the first candidate set. Thus, the first set is eliminated.

The second candidate set is $q_{N}=u_{q_{N}}=v^*-\frac{t}{2}$, $q_{NN}=q_{N}+\Delta q$, and $\tilde{p} \leq \min\{\Delta q, -5\Delta q \}$. Note that $\pi_G$ is concave with respect to $\Delta q$, and,

$$
\begin{aligned}
\frac{d q_{N}}{d \Delta q}&=0\\
\frac{d q_{NN}}{d \Delta q}&=1
\end{aligned}
$$

Solving $\frac{d\pi_G}{d\Delta q}=0$ yields $\Delta q^*=\tilde{p}+\frac{3}{2}t$. Note that $\Delta q^*>\tilde{p}$. Thus, $\Delta q^*$ is the optimum solution if $\tilde{p}\leq-5\Delta q^*$. Thus, $\Delta q^*\leq -\frac{\tilde{p}}{5}\Rightarrow \tilde{p}\leq -\frac{15t}{12}=-\frac{5t}{4}$ is the condition for optimality of this candidate set.

With similar computations, the third candidate set yields the optimum solution of $\Delta q^*=\tilde{p}-\frac{3}{2}t$.  The optimality condition of this set is $\Delta q^*\geq -\frac{\tilde{p}}{5}\Rightarrow \tilde{p}\geq \frac{5t}{4}$.

For the fourth set of candidate answers the optimum answer is $\Delta q^*=\tilde{p}+t$. However, $\Delta q^*>\tilde{p}$, which does not satisfy the condition for this set. Thus, the forth set is eliminated. 

Now, we prove Theorem~\ref{theorem:eq2}: 
\begin{proof}
We argued that in the first and the fourth sets of candidate answers the optimal solutions derived by the first order condition are not feasible. Therefore, since the payoff is concave, in these two sets, the candidate answers are on the boundaries.

We first consider the second and third sets of feasible answers. Note that in the second set of feasible answer $\tilde{p}\leq 0$, since $\tilde{p}\leq \min\{\Delta q,-5\Delta q\}$, and in the third set, $\tilde{p}\geq 0$, since $\tilde{p}\geq \max\{\Delta q,-5\Delta q\}$.

Thus if $\tilde{p}\leq 0$, the second candidate set should be considered. The condition  $ \tilde{p} \leq \min\{\Delta q, -5\Delta q\}$  is equivalent to $\tilde{p}\leq \Delta q \leq -\frac{\tilde{p}}{5}$. We argued that in the second set, if $\tilde{p}\leq -\frac{5t}{4}$, then $\Delta q^*=\tilde{p}+\frac{3t}{2}$. If not, then the optimum solution should be either $\Delta q^*=\tilde{p}$ or $\Delta q^*=-\frac{\tilde{p}}{5}$. By comparing the payoffs, we prove that in the second candidate set, when $-\frac{5t}{4}\leq \tilde{p}\leq 0$, $\Delta q^*=-\frac{\tilde{p}}{5}$.  

If $\Delta q=\tilde{p}$, in the second set:

\be
\ba
q_{N}=u_{q_{N}}&=v^*-\frac{t}{2}\\
q_{NN}=q_{N}+\Delta q&=v^*-\frac{t}{2}+\tilde{p}\\
\pi_G|_{\Delta q=\tilde{p}}&=v^*-\frac{t}{2}
\ea
\non
\ee
and if $\Delta q=-\frac{\tilde{p}}{5}$,

\be
\ba
q_{N}&=u_{q_{N}}=v^*-\frac{t}{2}\\
q_{NN}&=q_{N}+\Delta q=v^*-\frac{t}{2}-\frac{\tilde{p}}{5}\\
\pi_G|_{\Delta q=-\frac{\tilde{p}}{5}}&=(v^*-\frac{t}{2}-\frac{6\tilde{p}}{5})(\frac{1}{2}+\frac{\tilde{p}}{5t})+(v^*-\frac{t}{2})(\frac{1}{2}-\frac{\tilde{p}}{5t})\\
&=v^*-\frac{t}{2}-\frac{6\tilde{p}}{5}(\frac{1}{2}+\frac{\tilde{p}}{5t})
\ea
\non
\ee
Since $ -\frac{5t}{4} \leq \tilde{p}\leq 0$,  $\frac{1}{2}+\frac{\tilde{p}}{5t}>0$, and $\pi_G|_{\Delta q=-\frac{\tilde{p}}{5}}>\pi_G|_{\Delta q=\tilde{p}}$. Thus, in the second candidate set, when  $ -\frac{5t}{4} \leq \tilde{p}\leq 0$, $\Delta q^*=-\frac{\tilde{p}}{5}$. 

On the other hand, if $\tilde{p}\geq 0$, the third candidate set should be considered. The condition $ \tilde{p} \geq \max\{\Delta q, -5\Delta q\}$  is equivalent to $-\frac{\tilde{p}}{5} \leq \Delta q \leq \tilde{p}$. We argued that in the third set, if $\tilde{p}\geq \frac{5t}{4}$, then $\Delta q^*=\tilde{p}-\frac{3t}{2}$. If not, then the optimum solution should be either $\Delta q^*=\tilde{p}$ or $\Delta q^*=-\frac{\tilde{p}}{5}$. Again, by comparing the payoffs, we prove that in the third candidate set, when $0 \leq \tilde{p}\leq \frac{5t}{4}$, $\Delta q^*=-\frac{\tilde{p}}{5}$.

If $\Delta q=\tilde{p}$, in the third set:

\be
\ba
q_{NN}=u'_{q_{NN}}&=v^*-\frac{t}{2}\\
q_{N}=q_{NN}-\Delta q&=v^*-\frac{t}{2}-\tilde{p}\\
\pi_G|_{\Delta q=\tilde{p}}&=v^*-\frac{t}{2}-\tilde{p}
\ea
\non
\ee
and if $\Delta q=-\frac{\tilde{p}}{5}$,

\be
\ba
q_{NN}&=u'_{q_{NN}}=v^*-\frac{t}{2}\\
q_{N}&=q_{NN}-\Delta q=v^*-\frac{t}{2}+\frac{\tilde{p}}{5}\\
\pi_G|_{\Delta q=-\frac{\tilde{p}}{5}}&=(v^*-\frac{t}{2}-\tilde{p})(\frac{1}{2}+\frac{\tilde{p}}{5t})+(v^*-\frac{t}{2}+\frac{\tilde{p}}{5})(\frac{1}{2}-\frac{\tilde{p}}{5t})\\
&=v^*-\frac{t}{2}+\frac{\tilde{p}}{5}-\frac{6\tilde{p}}{5}(\frac{1}{2}+\frac{\tilde{p}}{5t})
\ea
\non
\ee

Since $ 0 \leq \tilde{p}\leq \frac{5t}{4}$,  $\frac{1}{2}+\frac{\tilde{p}}{5t}<1$, and $\pi_G|_{\Delta q=-\frac{\tilde{p}}{5}}>\pi_G|_{\Delta q=\tilde{p}}$. Thus, in the third candidate set, when  $ -\frac{5t}{4} \leq \tilde{p}\leq 0$, $\Delta q^*=-\frac{\tilde{p}}{5}$.

Now, consider  the first set. The condition $ -5\Delta q \leq  \tilde{p} \leq \Delta q$ is equivalent to $\Delta q \geq \max\{\tilde{p},-\frac{\tilde{p}}{5}\}$. Since the solution to the first order condition is not feasible, $\Delta q^*=\tilde{p}$ if $\tilde{p}\geq 0$, and $\Delta q^*=-\frac{\tilde{p}}{5}$ if $\tilde{p}\leq 0$. If $\tilde{p}\geq 0$, $q_{NN}=v^*-\frac{t}{2}$, and if $\tilde{p}\leq 0$, $q_{NN}=v^*-\frac{t}{2}-\frac{\tilde{p}}{5}$ and $q_{N}=v^*-\frac{t}{2}$. Both are two feasible points in the third and second candidate sets, respectively. Thus none of them can be an optimum solution.

The condition for the fourth case is $ \Delta q \leq \tilde{p} \leq -5 \Delta q$, which is equivalent to $\Delta q \leq \min \{\tilde{p},-\frac{\tilde{p}}{5}\}$. Since the solution to the first order condition is not feasible, $\Delta q^*=\tilde{p}$ if $\tilde{p}\leq 0$, and $\Delta q^*=-\frac{\tilde{p}}{5}$ if $\tilde{p}\geq 0$. If $\tilde{p}\leq 0$, $q_{N}=v^*-\frac{t}{2}$, and if $\tilde{p}\geq 0$, $q_{N}=v^*-\frac{t}{2}+\frac{\tilde{p}}{5}$ and $q_{NN}=v^*-\frac{t}{2}$. Both are two feasible points in the candidate second and third candidate sets, respectively. Thus none of them can be an optimum solution. The result follows.
\end{proof}

\end{document}